\documentclass[10pt,draft]{article}
\usepackage{amsfonts,amsmath,amssymb}
\DeclareMathOperator{\Real}{\textrm{Re}\,}

\newenvironment{corollary}{\vspace*{1ex}\newline\textbf{Corollary.}\em}{\newline}
\newenvironment{lemma}{\vspace*{1ex}\newline\textbf{Lemma 1.}\em}{\newline}
\newenvironment{proof}{\vspace*{1ex}\textbf{Proof.}}{\hfill$\Box$}
\begin{document}
\title{Green's function for the wavized Maxwell fish-eye problem}
\author{Rados{\l}aw Szmytkowski \\*[3ex]
Atomic Physics Division,
Department of Atomic Physics and Luminescence, \\
Faculty of Applied Physics and Mathematics,
Gda{\'n}sk University of Technology, \\
Narutowicza 11/12, 80--233 Gda{\'n}sk, Poland \\
email: radek@mif.pg.gda.pl}
\date{}
\maketitle
\begin{center}
\textbf{Published as: J.\ Phys.\ A 44 (2011) 065203} \\*[1ex]
\textbf{doi: 10.1088/1751-8113/44/6/065203} \\*[5ex]
\end{center}
\begin{abstract}
Unique transformation properties under the hyperspherical inversion
of a partial differential equation describing a stationary scalar
wave in an $N$-dimensional ($N\geqslant2$) Maxwell fish-eye medium
are exploited to construct a closed form of the Green's function for
that equation. For those wave numbers for which the Green's function
fails to exist, the generalized Green's function is derived.
Prospective physical applications are mentioned.
\end{abstract}
%
%\newpage
%
\section{Introduction}
\label{I}
\setcounter{equation}{0}
In 1854, Maxwell \cite{Maxw54} pointed at a remarkable property of an
infinite optical medium with the refraction index
\begin{equation}
n_{\mathrm{fe}}(\boldsymbol{r})
=\frac{2n_{0}\rho^{2}}{r^{2}+\rho^{2}}
\qquad (\rho>0).
\label{1.1}
\end{equation}
Within the framework of the geometrical optics, he proved that paths
of all rays emitted from an arbitrarily located point
$\boldsymbol{r}'$ are circles having two points in common: the source
point $\boldsymbol{r}'$ and the image point
$-\rho^{2}\boldsymbol{r}'/r^{\prime\,2}$. For this medium Maxwell
coined the name `the fish-eye'. In 1926, Carath{\'e}odory
\cite{Cara26} observed that there is a geometrical correspondence
between the circular rays in the fish-eye medium and the geodesics on
a sphere. This geometric thread was pursued further by other
researchers within the group-theoretical framework (cf, e.g., Ref.\
\cite{Wolf04}).

It seems that Demkov and Ostrovsky \cite{Demk71} were the first to
discuss the wavized scalar fish-eye problem. Specifically, they
considered the equation
\begin{equation}
\left[\boldsymbol{\nabla}^{2}
+\frac{4\nu(\nu+1)\rho^{2}}{(r^{2}+\rho^{2})^{2}}\right]
\Psi(\boldsymbol{r})=0
\qquad (\rho>0)
\label{1.2}
\end{equation}
in $\mathbb{R}^{3}$, subject to the boundary condition that
$\Psi(\boldsymbol{r})$ vanishes at infinity. They proceeded in two
directions. First, they solved analytically a spectral problem with
$\nu$ being an eigenparameter and showed that the resulting spectrum
is purely discrete and eigenfunctions may be expressed in terms of
the Gegenbauer polynomials. Second, they proved that Eq.\ (\ref{1.2})
possesses a certain remarkable transformation property under the
geometrical inversion in a certain class of spheres, and ingeniously
exploited this fact to construct a closed form of the relevant
Green's function in $\mathbb{R}^{3}$. Group-theoretical properties of
the scalar fish-eye wave equation (\ref{1.2}) were then investigated
in $\mathbb{R}^{2}$ by Frank \emph{et al.\/} \cite{Fran90,Fran91}.
Lately, the two-dimensional fish-eye medium has been studied by
Makowski and G{\'o}rska \cite{Mako09} in the context of the
construction of pertinent coherent states. Finally, in two very
recent papers, Leonhardt \cite{Leon09a} and Leonhardt and Philbin
\cite{Leon09b} have argued that the geometric-optical perfect
focusing property of the fish-eye medium, discovered by Maxwell,
holds as well within the wave-optics framework; that issue will be
critically reexamined in our upcoming work, with the aid of the
results presented below.

The present paper is the first out of a series of several reports in
which we shall expose results of our research on the wave properties
of the Maxwell fish-eye and related media. Here, we derive a closed
form of the Green's function for the scalar fish-eye wave equation in
$\mathbb{R}^{N}$, $N\geqslant2$. The particular method we employ
generalizes the aforementioned one used by Demkov and Ostrovsky in
the case of $N=3$ and exploits a peculiar transformation property of
the $N$-dimensional fish-eye equation under the hyperspherical
inversion.

The structure of the paper is as follows. In Section \ref{II}, we
investigate transformation properties of a class of partial
differential equations under the hyperspherical inversion, with a
special focus on the $N$-dimensional fish-eye equation. The results
of that investigation are used in Section \ref{III} to construct the
Green's function for the fish-eye problem. The case when the Green's
function fails to exist and is to be replaced by the generalized
Green's function is considered in Section \ref{IV}. Prospective
physical applications of the results are briefly discussed in Section
\ref{V}. The paper ends with an appendix, in which a number of
closed-form representations of the derivative $[\partial
P_{\nu}^{-N/2+1}(x)/ \partial\nu]_{\nu=n+N/2-1}$, with $x\in(-1,1)$,
$N\in\mathbb{N}\setminus\{0,1\}$ and $n\in\mathbb{N}$, required in
Section \ref{IV}, are displayed. The list of references attached has
been intended to contain representative items rather than to be a
comprehensive one. An exhaustive listing of works relevant to the
Maxwell fish-eye problem will be included in one of our forthcoming
papers.
%
%\newpage
%
\section{Transformation properties of a class of partial differential
equations under the hyperspherical inversion}
\label{II}
\setcounter{equation}{0}
At first, we establish the following
\begin{lemma}
If $\Psi(\boldsymbol{r})$ \emph{(}with
$\boldsymbol{r}\in\mathbb{R}^{N}$, $N\geqslant2$\emph{)} satisfies
the equation
\begin{equation}
\left[\boldsymbol{\nabla}^{2}+k^{2}n^{2}(\boldsymbol{r})\right]
\Psi(\boldsymbol{r})=0,
\label{2.1}
\end{equation}
then for arbitrary $R\in\mathbb{R}$ and
$\boldsymbol{a},\boldsymbol{b}\in\mathbb{R}^{N}$ it holds that
\begin{equation}
\left[\boldsymbol{\nabla}^{2}
+\frac{k^{2}R^{4}}{|\boldsymbol{r}-\boldsymbol{a}|^{4}}
n^{2}\left(\frac{R^{2}}{|\boldsymbol{r}-\boldsymbol{a}|^{2}}
(\boldsymbol{r}-\boldsymbol{a})+\boldsymbol{b}\right)\right]
\frac{1}{|\boldsymbol{r}-\boldsymbol{a}|^{N-2}}
\Psi\left(\frac{R^{2}}{|\boldsymbol{r}-\boldsymbol{a}|^{2}}
(\boldsymbol{r}-\boldsymbol{a})+\boldsymbol{b}\right)=0.
\label{2.2}
\end{equation}
\end{lemma}
\begin{proof}
The transformation
\begin{equation}
\boldsymbol{r}\mapsto\frac{R^{2}}{|\boldsymbol{r}-\boldsymbol{a}|^{2}}
(\boldsymbol{r}-\boldsymbol{a})+\boldsymbol{b}
\label{2.3}
\end{equation}
results in the following alteration of the infinitesimal line
element:
\begin{equation}
(\mathrm{d}\boldsymbol{r})^{2}
\mapsto\frac{R^{4}}{|\boldsymbol{r}-\boldsymbol{a}|^{4}}
(\mathrm{d}\boldsymbol{r})^{2}.
\label{2.4}
\end{equation}
Hence, substitution (\ref{2.3}) implies the following transformation
of the $N$-dimensional Laplace operator:
\begin{equation}
\boldsymbol{\nabla}^{2}
\mapsto\frac{|\boldsymbol{r}-\boldsymbol{a}|^{2N}}{R^{2N}}
\boldsymbol{\nabla}\cdot\left(\frac{R^{2(N-2)}}
{|\boldsymbol{r}-\boldsymbol{a}|^{2(N-2)}}\boldsymbol{\nabla}\right).
\label{2.5}
\end{equation}
Using the easily provable differential identity
\begin{equation}
\boldsymbol{\nabla}\cdot\left(\frac{R^{2(N-2)}}
{|\boldsymbol{r}-\boldsymbol{a}|^{2(N-2)}}\boldsymbol{\nabla}\right)
=\frac{R^{2(N-2)}}{|\boldsymbol{r}-\boldsymbol{a}|^{N-2}}
\boldsymbol{\nabla}^{2}
\frac{1}{|\boldsymbol{r}-\boldsymbol{a}|^{N-2}}
-\frac{R^{2(N-2)}}{|\boldsymbol{r}-\boldsymbol{a}|^{N-2}}
\left(\boldsymbol{\nabla}^{2}
\frac{1}{|\boldsymbol{r}-\boldsymbol{a}|^{N-2}}\right),
\label{2.6}
\end{equation}
and exploiting the fact that the function
$1/|\boldsymbol{r}-\boldsymbol{a}|^{N-2}$ is harmonic in
$\mathbb{R}^{N}$, we see that transformation (\ref{2.3}) changes Eq.\
(\ref{2.1}) into
\begin{equation}
\left[\frac{|\boldsymbol{r}-\boldsymbol{a}|^{N+2}}{R^{4}}
\boldsymbol{\nabla}^{2}
\frac{1}{|\boldsymbol{r}-\boldsymbol{a}|^{N-2}}
+k^{2}n^{2}\left(\frac{R^{2}}{|\boldsymbol{r}-\boldsymbol{a}|^{2}}
(\boldsymbol{r}-\boldsymbol{a})+\boldsymbol{b}\right)\right]
\Psi\left(\frac{R^{2}}{|\boldsymbol{r}-\boldsymbol{a}|^{2}}
(\boldsymbol{r}-\boldsymbol{a})+\boldsymbol{b}\right)=0,
\label{2.7}
\end{equation}
which immediately leads to Eq.\ (\ref{2.2}).
\end{proof}

Actually, the above lemma offers a bit more than necessary for the
purposes of this paper. In view of our needs, in what follows we
shall restrict ourselves to the special case when the vectors
$\boldsymbol{a}$ and $\boldsymbol{b}$ are equal. It is then evident
that the resulting transformation
\begin{equation}
\boldsymbol{r}\mapsto\frac{R^{2}}{|\boldsymbol{r}-\boldsymbol{a}|^{2}}
(\boldsymbol{r}-\boldsymbol{a})+\boldsymbol{a}
\label{2.8}
\end{equation}
is the geometric inversion in the hypersphere of radius $R$ centered
at the point with the radius vector $\boldsymbol{r}=\boldsymbol{a}$.
(In fact, if $\boldsymbol{b}=\boldsymbol{a}$ and $k=0$, the lemma is
simply an $N$-dimensional extension of the well-known Kelvin
inversion theorem for harmonic functions \cite{Thom72}.)

Now we turn to the fish-eye problem. Application of the following
special case of inversion (\ref{2.8}):
\begin{equation}
\boldsymbol{r}
\mapsto\frac{\rho^{2}+a^{2}}{|\boldsymbol{r}-\boldsymbol{a}|^{2}}
(\boldsymbol{r}-\boldsymbol{a})+\boldsymbol{a}
\quad \Rightarrow \quad
r\mapsto\frac{a}{|\boldsymbol{r}-\boldsymbol{a}|}
\left|\boldsymbol{r}+\boldsymbol{a}\frac{\rho^{2}}{a^{2}}\right|
\label{2.9}
\end{equation}
to the fish-eye refraction index (\ref{1.1}) gives
\begin{equation}
n_{\mathrm{fe}}\left(\frac{\rho^{2}+a^{2}}
{|\boldsymbol{r}-\boldsymbol{a}|^{2}}(\boldsymbol{r}-\boldsymbol{a})
+\boldsymbol{a}\right)=\frac{|\boldsymbol{r}-\boldsymbol{a}|^{2}}
{\rho^{2}+a^{2}}\frac{2n_{0}\rho^{2}}{r^{2}+\rho^{2}}
=\frac{|\boldsymbol{r}-\boldsymbol{a}|^{2}}{\rho^{2}+a^{2}}
n_{\mathrm{fe}}(\boldsymbol{r}).
\label{2.10}
\end{equation}
Combining this property of the index
$n_{\mathrm{fe}}(\boldsymbol{r})$ with the result stated in the
lemma, we arrive at
\begin{corollary}
If the function $\Psi(\boldsymbol{r})$ \emph{(}with 
$\boldsymbol{r}\in\mathbb{R}^{N}$, $N\geqslant2$\emph{)} solves the
fish-eye equation
\begin{equation}
\left[\boldsymbol{\nabla}^{2}
+\frac{4n_{0}^{2}k^{2}\rho^{4}}{(r^{2}+\rho^{2})^{2}}\right]
\Psi(\boldsymbol{r})=0
\qquad (\rho>0),
\label{2.11}
\end{equation}
then for arbitrary $\boldsymbol{a}\in\mathbb{R}^{N}$ the function
\begin{equation}
\hat{\mathcal{I}}(\boldsymbol{a},\sqrt{\rho^{2}+a^{2}})
{\Psi}(\boldsymbol{r})
\equiv\frac{1}{|\boldsymbol{r}-\boldsymbol{a}|^{N-2}}
\Psi\left(\frac{\rho^{2}+a^{2}}{|\boldsymbol{r}-\boldsymbol{a}|^{2}}
(\boldsymbol{r}-\boldsymbol{a})+\boldsymbol{a}\right)
\label{2.12}
\end{equation}
also solves this equation, i.e., it holds that
\begin{equation}
\left[\boldsymbol{\nabla}^{2}
+\frac{4n_{0}^{2}k^{2}\rho^{4}}{(r^{2}+\rho^{2})^{2}}\right]
\frac{1}{|\boldsymbol{r}-\boldsymbol{a}|^{N-2}}
\Psi\left(\frac{\rho^{2}+a^{2}}{|\boldsymbol{r}-\boldsymbol{a}|^{2}}
(\boldsymbol{r}-\boldsymbol{a})+\boldsymbol{a}\right)=0
\qquad (\rho>0).
\label{2.13}
\end{equation}
\end{corollary}

In the particular case of $N=3$, the above result was established by
Demkov and Ostrovsky \cite{Demk71} (see also \cite{Demk84}).
\section{The Green's function for the fish-eye problem}
\label{III}
\setcounter{equation}{0}
We are now ready to construct the $N$-dimensional fish-eye Green's
function $G_{\nu}(\boldsymbol{r},\boldsymbol{r}')$. According to the
general theory of Green's functions for elliptic partial differential
operators, $G_{\nu}(\boldsymbol{r},\boldsymbol{r}')$ is a
single-valued solution to the fish-eye equation
\begin{equation}
\left[\boldsymbol{\nabla}^{2}
+\frac{4n_{0}^{2}k^{2}\rho^{4}}{(r^{2}+\rho^{2})^{2}}\right]
G_{\nu}(\boldsymbol{r},\boldsymbol{r}')=0
\qquad (\rho>0)
\label{3.1}
\end{equation}
everywhere in $\mathbb{R}^{N}$ except for the source point
$\boldsymbol{r}=\boldsymbol{r}'$, where it diverges according to
\begin{equation}
G_{\nu}(\boldsymbol{r},\boldsymbol{r}')
\stackrel{\boldsymbol{r}\to\boldsymbol{r}'}
{\longrightarrow}\frac{1}{2\pi}\ln|\boldsymbol{r}-\boldsymbol{r}'|
\qquad (N=2)
\label{3.2}
\end{equation}
or
\begin{equation}
G_{\nu}(\boldsymbol{r},\boldsymbol{r}')
\stackrel{\boldsymbol{r}\to\boldsymbol{r}'}{\longrightarrow}
-\frac{1}{(N-2)S_{N-1}|\boldsymbol{r}-\boldsymbol{r}'|^{N-2}}
\qquad (N\geqslant3).
\label{3.3}
\end{equation}
In the last equation 
\begin{equation}
S_{N-1}=\frac{2\pi^{N/2}}{\Gamma\left(\frac{N}{2}\right)}
\label{3.4}
\end{equation}
is a surface area of a unit $(N-1)$-dimensional sphere
$\mathbb{S}^{N-1}$ embedded in $\mathbb{R}^{N}$. At infinity, we
require
\begin{equation}
G_{\nu}(\boldsymbol{r},\boldsymbol{r}')
\stackrel{r\to\infty}{\longrightarrow}
\frac{C_{\nu}(\boldsymbol{r}')}{r^{N-2}}
\label{3.5}
\end{equation}
(the non-zero constant $C_{\nu}(\boldsymbol{r}')$ appearing in
condition (\ref{3.5}) will be determined later). The parameter $\nu$
is defined as
\begin{equation}
\nu=\frac{-1+\sqrt{1+4n_{0}^{2}k^{2}\rho^{2}}}{2}
\qquad (\textrm{$\nu\to0$ for $n_{0}\to0$})
\label{3.6}
\end{equation}
and reasons for its introduction will become clear shortly. In what
follows, we admit that the product $n_{0}^{2}k^{2}$, hence also
$\nu$, may be complex.

At first, consider the case when the source is located at the center
of symmetry of the medium. Evidently, the corresponding Green's
function $G_{\nu}(\boldsymbol{r},\boldsymbol{0})$ must be spherically
symmetric, being a function of $r=|\boldsymbol{r}|$ only. Hence, it
follows that $G_{\nu}(\boldsymbol{r},\boldsymbol{0})$ obeys
\begin{equation}
\left[\frac{\partial^{2}}{\partial r^{2}}
+\frac{N-1}{r}\frac{\partial}{\partial r}
+\frac{4\nu(\nu+1)\rho^{2}}{(r^{2}+\rho^{2})^{2}}\right]
G_{\nu}(\boldsymbol{r},\boldsymbol{0})=0
\label{3.7}
\end{equation}
except for the point $\boldsymbol{r}=\boldsymbol{0}$, where it
behaves according to Eqs.\ (\ref{3.2}) or (\ref{3.3}) with
$\boldsymbol{r}'=\boldsymbol{0}$. The substitution
\begin{equation}
G_{\nu}(\boldsymbol{r},\boldsymbol{0})
=\left(\frac{\rho}{r}\right)^{N/2-1}
F\left(\frac{r^{2}-\rho^{2}}{r^{2}+\rho^{2}}\right)
\label{3.8}
\end{equation}
leads to the following differential equation for the function $F$:
\begin{equation}
\left[(1-x^{2})\frac{\mathrm{d}^{2}}{\mathrm{d}x^{2}}
-2x\frac{\mathrm{d}}{\mathrm{d}x}+\nu(\nu+1)
-\frac{\mu^{2}}{1-x^{2}}\right]F(x)=0,
\label{3.9}
\end{equation}
where
\begin{equation}
x=\frac{r^{2}-\rho^{2}}{r^{2}+\rho^{2}}
\qquad (-1\leqslant x\leqslant1)
\label{3.10}
\end{equation}
and
\begin{equation}
\mu=\frac{N-2}{2}.
\label{3.11}
\end{equation}

Equation (\ref{3.9}) is the associated Legendre equation. Its general
solution, written in the form most suitable for the present purposes,
is
\begin{equation}
F(x)=AP_{\nu}^{-\mu}(x)+BR_{\nu}^{\mu}(x),
\label{3.12}
\end{equation}
with $A$, $B$ being arbitrary constants and with
\begin{equation}
R_{\nu}^{\mu}(x)=Q_{\nu}^{\mu}(x)
+\frac{\mathrm{i}\pi}{2}P_{\nu}^{\mu}(x)
=\frac{\pi}{2\sin(\pi\mu)}
\left[\mathrm{e}^{\mathrm{i}\pi\mu}P_{\nu}^{\mu}(x)
-\frac{\Gamma(\nu+\mu+1)}{\Gamma(\nu-\mu+1)}P_{\nu}^{-\mu}(x)\right].
\label{3.13}
\end{equation}
Here, $P_{\nu}^{\mu}(x)$ and $Q_{\nu}^{\mu}(x)$ are the associated
Legendre functions (on the cut $-1\leqslant x\leqslant1$) of the
first and second kinds, respectively (occasionally,
$R_{\nu}^{\mu}(x)$ is called the associated Legendre function of the
third kind). The general character of solution (\ref{3.12}) follows
from the fact that the Wronskian of $P_{\nu}^{-\mu}(x)$ and
$R_{\nu}^{\mu}(x)$ is
\begin{equation}
W[P_{\nu}^{-\mu}(x),R_{\nu}^{\mu}(x)]
=\frac{\exp(\mathrm{i}\pi\mu)}{1-x^{2}},
\label{3.14}
\end{equation}
i.e., it vanishes nowhere. Henceforth, we shall adopt the standard
convention and shall write $P_{\nu}(x)$ and $R_{\nu}(x)$ in place of
$P_{\nu}^{0}(x)$ and $R_{\nu}^{0}(x)$.

With the general solution to Eq.\ (\ref{3.9}) in hand, we see that
the Green's function $G_{\nu}(\boldsymbol{r},\boldsymbol{0})$ is of the
form
\begin{equation}
G_{\nu}(\boldsymbol{r},\boldsymbol{0})
=A\left(\frac{\rho}{r}\right)^{N/2-1}
P_{\nu}^{-N/2+1}\left(\frac{r^{2}-\rho^{2}}{r^{2}+\rho^{2}}\right)
+B\left(\frac{\rho}{r}\right)^{N/2-1}
R_{\nu}^{N/2-1}\left(\frac{r^{2}-\rho^{2}}{r^{2}+\rho^{2}}\right),
\label{3.15}
\end{equation}
$\nu$ being defined in Eq.\ (\ref{3.6}). We shall fix values of the
constants $A$ and $B$ in two steps. At first, we investigate the
asymptotics of the expression in Eq.\ (\ref{3.15}) as $r\to\infty$.
Using the known formulas \cite[p.~196]{Magn66}
\begin{equation}
P_{\nu}^{-\mu}(x)\stackrel{x\to1-0}{\longrightarrow}
\frac{1}{\Gamma(\mu+1)}\left(\frac{1-x}{2}\right)^{\mu/2}
\qquad (\mu\neq-1,-2,\ldots),
\label{3.16}
\end{equation}
\begin{equation}
R_{\nu}(x)\stackrel{x\to1-0}{\longrightarrow}
-\frac{1}{2}\ln(1-x),
\label{3.17}
\end{equation}
\begin{equation}
R_{\nu}^{\mu}(x)\stackrel{x\to1-0}{\longrightarrow}
\frac{1}{2}\mathrm{e}^{\mathrm{i}\pi\mu}\Gamma(\mu)
\left(\frac{1-x}{2}\right)^{-\mu/2}
\qquad (\Real\mu>0),
\label{3.18}
\end{equation}
we see that the constraint (\ref{3.5}) is fulfilled iff $B=0$, and
consequently
\begin{equation}
G_{\nu}(\boldsymbol{r},\boldsymbol{0})
=A\left(\frac{\rho}{r}\right)^{N/2-1}
P_{\nu}^{-N/2+1}\left(\frac{r^{2}-\rho^{2}}{r^{2}+\rho^{2}}\right).
\label{3.19}
\end{equation}
In the second step, we investigate the asymptotics of the right-hand
side of Eq.\ (\ref{3.19}) as $\boldsymbol{r}\to\boldsymbol{0}$.
Exploiting the formulas \cite[p.~197]{Magn66}
\begin{equation}
P_{\nu}(x)\stackrel{x\to-1+0}{\longrightarrow}
\frac{\sin(\pi\nu)}{\pi}\ln(1+x)
\label{3.20}
\end{equation}
and
\begin{equation}
P_{\nu}^{-\mu}(x)\stackrel{x\to-1+0}{\longrightarrow}
\frac{\Gamma(\mu)}{\Gamma(\nu+\mu+1)\Gamma(-\nu+\mu)}
\left(\frac{1+x}{2}\right)^{-\mu/2}
\qquad (\Real\mu>0),
\label{3.21}
\end{equation}
we find that the constraints (\ref{3.2}) and (\ref{3.3}) will be
satisfied iff
\begin{equation}
G_{\nu}(\boldsymbol{r},\boldsymbol{0})=\frac{1}{4\sin(\pi\nu)}
P_{\nu}\left(\frac{r^{2}-\rho^{2}}{r^{2}+\rho^{2}}\right)
\qquad (N=2)
\label{3.22}
\end{equation}
and
\begin{equation}
G_{\nu}(\boldsymbol{r},\boldsymbol{0})
=-\frac{\Gamma\left(\frac{N}{2}+\nu\right)
\Gamma\left(\frac{N}{2}-\nu-1\right)}{4\pi^{N/2}}
\frac{\displaystyle P_{\nu}^{-N/2+1}
\left(\frac{r^{2}-\rho^{2}}{r^{2}+\rho^{2}}\right)}
{(r\rho)^{N/2-1}}
\qquad (N\geqslant3),
\label{3.23}
\end{equation}
respectively. Since it holds that
\begin{equation}
\sin(\pi\nu)=-\frac{\pi}{\Gamma(\nu+1)\Gamma(-\nu)},
\label{3.24}
\end{equation}
Eqs.\ (\ref{3.22}) and (\ref{3.23}) may be collected into a single
formula
\begin{equation}
G_{\nu}(\boldsymbol{r},\boldsymbol{0})
=-\frac{\Gamma\left(\frac{N}{2}+\nu\right)
\Gamma\left(\frac{N}{2}-\nu-1\right)}{4\pi^{N/2}}
\frac{\displaystyle P_{\nu}^{-N/2+1}
\left(\frac{r^{2}-\rho^{2}}{r^{2}+\rho^{2}}\right)}
{(r\rho)^{N/2-1}}.
\label{3.25}
\end{equation}
From this, using relation (\ref{3.16}), we deduce that for
$\boldsymbol{r}'=\boldsymbol{0}$ the constant $C_{\nu}$ in the
asymptotic relation (\ref{3.5}) is
\begin{equation}
C_{\nu}(\boldsymbol{0})=-\frac{\Gamma\left(\frac{N}{2}+\nu\right)
\Gamma\left(\frac{N}{2}-\nu-1\right)}
{4\pi^{N/2}\Gamma\left(\frac{N}{2}\right)}.
\label{3.26} 
\end{equation}

To find the Green's function for an arbitrary location of the source
point $\boldsymbol{r}'$, we consider the transformed function
\begin{equation}
\hat{\mathcal{I}}(\boldsymbol{a},\sqrt{\rho^{2}+a^{2}})
G_{\nu}(\boldsymbol{r},\boldsymbol{0})
=\frac{1}{|\boldsymbol{r}-\boldsymbol{a}|^{N-2}}
G_{\nu}\left(\frac{\rho^{2}+a^{2}}{|\boldsymbol{r}-\boldsymbol{a}|^{2}}
(\boldsymbol{r}-\boldsymbol{a})+\boldsymbol{a},\boldsymbol{0}\right),
\label{3.27}
\end{equation}
with the center of the inversion sphere (of radius
$\sqrt{\rho^{2}+a^{2}}$) located at the point
\begin{equation}
\boldsymbol{a}=-\boldsymbol{r}'\frac{\rho^{2}}{r^{\prime\,2}}.
\label{3.28}
\end{equation}
Using Eq.\ (\ref{3.25}), the explicit form of this transformed
function, denoted hereafter as 
$g_{\nu}(\boldsymbol{r},\boldsymbol{r}')$, is seen to be
\begin{equation}
g_{\nu}(\boldsymbol{r},\boldsymbol{r}')
=-\,\frac{\Gamma\left(\frac{N}{2}+\nu\right)
\Gamma\left(\frac{N}{2}-\nu-1\right)}{4\pi^{N/2}}
\left(\frac{r'}{\rho^{3}}\right)^{N/2-1}
\frac{\displaystyle P_{\nu}^{-N/2+1}
\left(-1+\frac{\displaystyle2\rho^{2}
(\boldsymbol{r}-\boldsymbol{r}')^{2}}
{(r^{2}+\rho^{2})(r^{\prime\,2}+\rho^{2})}\right)}
{\displaystyle|\boldsymbol{r}-\boldsymbol{r}'|^{N/2-1}
\left|\boldsymbol{r}+\boldsymbol{r}'\frac{\rho^{2}}{r^{\prime\,2}}
\right|^{N/2-1}}.
\label{3.29}
\end{equation}
In view of the results of Section \ref{II}, we know for sure that the
function in Eq.\ (\ref{3.29}) solves the fish-eye equation, except,
possibly, for some isolated points. It is evident that the points at
which the behavior of $g_{\nu}(\boldsymbol{r},\boldsymbol{r}')$
should be investigated are the two finite points
$\boldsymbol{r}=\boldsymbol{r}'$ and
$\boldsymbol{r}=-\boldsymbol{r}'\rho^{2}/r^{\prime\,2}$, and also the
point at infinity. Using the asymptotic relations (\ref{3.20}) and
(\ref{3.21}), we derive
\begin{equation}
g_{\nu}(\boldsymbol{r},\boldsymbol{r}')
\stackrel{\boldsymbol{r}\to\boldsymbol{r}'}{\longrightarrow}
\frac{1}{2\pi}\ln|\boldsymbol{r}-\boldsymbol{r}'|
\qquad (N=2)
\label{3.30}
\end{equation}
and
\begin{equation}
g_{\nu}(\boldsymbol{r},\boldsymbol{r}')
\stackrel{\boldsymbol{r}\to\boldsymbol{r}'}{\longrightarrow}
-\frac{1}{(N-2)S_{N-1}|\boldsymbol{r}-\boldsymbol{r}'|^{N-2}}
\left(\frac{r'}{\rho^{2}}\right)^{N-2}
\qquad (N\geqslant3),
\label{3.31}
\end{equation}
i.e., the `inverted' function diverges for
$\boldsymbol{r}\to\boldsymbol{r}'$ in the same manner (save for the
factor $(r'/\rho^{2})^{N-2}$ when $N\geqslant3$) as the Green's
function $G_{\nu}(\boldsymbol{r},\boldsymbol{r}')$ (cf.\ Eqs.\
(\ref{3.2}) and (\ref{3.3})). Furthermore, it is seen that for
$r\to\infty$ the function (\ref{3.29}) decays asymptotically as
\begin{equation}
g_{\nu}(\boldsymbol{r},\boldsymbol{r}')
\stackrel{r\to\infty}{\longrightarrow}
-\frac{\Gamma\left(\frac{N}{2}+\nu\right)
\Gamma\left(\frac{N}{2}-\nu-1\right)}{4\pi^{N/2}}
\left(\frac{r'}{\rho^{3}}\right)^{N/2-1}
\frac{\displaystyle P_{\nu}^{-N/2+1}
\left(\frac{\rho^{2}-r^{\prime\,2}}{\rho^{2}+r^{\prime\,2}}\right)}
{r^{N-2}},
\label{3.32}
\end{equation}
i.e., in the same functional manner with $r$ as prescribed for
$G_{\nu}(\boldsymbol{r},\boldsymbol{r}')$ in Eq.\ (\ref{3.5}).
Finally, with the help of the identity
\begin{equation}
-1+\frac{2\rho^{2}(\boldsymbol{r}-\boldsymbol{r}')^{2}}
{(r^{2}+\rho^{2})(r^{\prime\,2}+\rho^{2})}
=1-\frac{\displaystyle2r^{\prime\,2}
\left(\boldsymbol{r}+\boldsymbol{r}'
\frac{\rho^{2}}{r^{\prime\,2}}\right)^{2}}
{(r^{2}+\rho^{2})(r^{\prime\,2}+\rho^{2})}
\label{3.33}
\end{equation}
and the asymptotic relation (\ref{3.16}), we find that for
$\boldsymbol{r}\to-\boldsymbol{r}'\rho^{2}/r^{\prime\,2}$ the
function (\ref{3.29}) remains finite:
\begin{equation}
g_{\nu}(\boldsymbol{r},\boldsymbol{r}')
\stackrel{\boldsymbol{r}\to-\boldsymbol{r}'\rho^{2}/r^{\prime\,2}}
{\longrightarrow}-\frac{\Gamma\left(\frac{N}{2}+\nu\right)
\Gamma\left(\frac{N}{2}-\nu-1\right)}
{4\pi^{N/2}\Gamma\left(\frac{N}{2}\right)}
\left(\frac{r'}{\rho}\right)^{2(N-2)}
\frac{1}{(r^{\prime\,2}+\rho^{2})^{N-2}}.
\label{3.34}
\end{equation}
Thus, we see that the function $(\rho^{2}/r')^{N-2}
g_{\nu}(\boldsymbol{r},\boldsymbol{r}')$ satisfies all conditions
imposed on $G_{\nu}(\boldsymbol{r},\boldsymbol{r}')$ in Eqs.\
(\ref{3.1})--(\ref{3.5}). Hence, we conclude that the closed form of
the $N$-dimensional fish-eye Green's function for an arbitrary location
of the source point $\boldsymbol{r}'$ is
\begin{equation}
G_{\nu}(\boldsymbol{r},\boldsymbol{r}')
=-\frac{\Gamma\left(\frac{N}{2}+\nu\right)
\Gamma\left(\frac{N}{2}-\nu-1\right)}{4\pi^{N/2}}
\left(\frac{\rho}{r'}\right)^{N/2-1}
\frac{\displaystyle P_{\nu}^{-N/2+1}
\left(-1+\frac{\displaystyle2\rho^{2}
(\boldsymbol{r}-\boldsymbol{r}')^{2}}
{(r^{2}+\rho^{2})(r^{\prime\,2}+\rho^{2})}\right)}
{\displaystyle|\boldsymbol{r}-\boldsymbol{r}'|^{N/2-1}
\left|\boldsymbol{r}+\boldsymbol{r}'\frac{\rho^{2}}{r^{\prime\,2}}
\right|^{N/2-1}}
\label{3.35}
\end{equation}
and that the constant $C_{\nu}(\boldsymbol{r}')$ in the asymptotic
constraint (\ref{3.5}) is
\begin{equation}
C_{\nu}(\boldsymbol{r}')=-\frac{\Gamma\left(\frac{N}{2}+\nu\right)
\Gamma\left(\frac{N}{2}-\nu-1\right)}{4\pi^{N/2}}
\left(\frac{\rho}{r'}\right)^{N/2-1}
P_{\nu}^{-N/2+1}\left(\frac{\rho^{2}-r^{\prime\,2}}
{\rho^{2}+r^{\prime\,2}}\right).
\label{3.36}
\end{equation}

Since the differential operator in Eq.\ (\ref{3.1}) is symmetric with
respect to the scalar product $\big<\chi\big|\phi\big>_{N}
\equiv\int_{\mathbb{R}^{N}}\mathrm{d}^{N}\boldsymbol{r}\:
\chi(\boldsymbol{r})\phi(\boldsymbol{r})$, the fish-eye Green's function
should be symmetric with respect to the interchange of the source and
observation points:
\begin{equation}
G_{\nu}(\boldsymbol{r},\boldsymbol{r}')
=G_{\nu}(\boldsymbol{r}',\boldsymbol{r}).
\label{3.37}
\end{equation}
However, the representation of
$G_{\nu}(\boldsymbol{r},\boldsymbol{r}')$ given in Eq.\ (\ref{3.35})
does not exhibit this property explicitly. To show that nevertheless
relation (\ref{3.37}) is satisfied, we observe that it holds that
\begin{equation}
r'\left|\boldsymbol{r}+\boldsymbol{r}'\frac{\rho^{2}}
{r^{\prime\,2}}\right|
=r\left|\boldsymbol{r}'+\boldsymbol{r}\frac{\rho^{2}}{r^{2}}\right|
=\sqrt{r^{2}r^{\prime\,2}
+2\rho^{2}\boldsymbol{r}\cdot\boldsymbol{r}'+\rho^{4}}.
\label{3.38}
\end{equation}
Consequently, the Green's function (\ref{3.35}) may be alternatively
rewritten in either of the following two manifestly symmetric forms:
\begin{equation}
G_{\nu}(\boldsymbol{r},\boldsymbol{r}')
=-\frac{\Gamma\left(\frac{N}{2}+\nu\right)
\Gamma\left(\frac{N}{2}-\nu-1\right)}{4\pi^{N/2}}
\frac{\displaystyle\rho^{N/2-1}
P_{\nu}^{-N/2+1}\left(-1+\frac{\displaystyle2\rho^{2}
(\boldsymbol{r}-\boldsymbol{r}')^{2}}
{(r^{2}+\rho^{2})(r^{\prime\,2}+\rho^{2})}\right)}
{\displaystyle|\boldsymbol{r}-\boldsymbol{r}'|^{N/2-1}
\left(r^{2}r^{\prime\,2}+2\rho^{2}\boldsymbol{r}\cdot\boldsymbol{r}'
+\rho^{4}\right)^{N/4-1/2}}
\label{3.39}
\end{equation}
or
\begin{eqnarray}
G_{\nu}(\boldsymbol{r},\boldsymbol{r}')
&=& -\,\frac{\Gamma\left(\frac{N}{2}+\nu\right)
\Gamma\left(\frac{N}{2}-\nu-1\right)}{4\pi^{N/2}}
\nonumber \\
&& \times\frac{\displaystyle\rho^{N/2-1}
P_{\nu}^{-N/2+1}\left(-1+\frac{\displaystyle2\rho^{2}
(\boldsymbol{r}-\boldsymbol{r}')^{2}}
{(r^{2}+\rho^{2})(r^{\prime\,2}+\rho^{2})}\right)}
{\displaystyle(rr')^{N/4-1/2}|
\boldsymbol{r}-\boldsymbol{r}'|^{N/2-1}
\left|\boldsymbol{r}+\boldsymbol{r}'\frac{\rho^{2}}{r^{\prime\,2}}
\right|^{N/4-1/2}
\left|\boldsymbol{r}'+\boldsymbol{r}\frac{\rho^{2}}{r^{2}}
\right|^{N/4-1/2}}.
\label{3.40}
\end{eqnarray}
Still another representation of
$G_{\nu}(\boldsymbol{r},\boldsymbol{r}')$ displaying its symmetry is
the one in terms of the Gegenbauer function of the first kind. Using
the known relationship
\begin{equation}
C_{\alpha}^{\lambda}(x)=\frac{\sqrt{\pi}}{2^{\lambda-1/2}}
\frac{\Gamma(\alpha+2\lambda)}{\Gamma(\lambda)\Gamma(\alpha+1)}
(1-x^{2})^{-\lambda/2+1/4}P_{\alpha+\lambda-1/2}^{-\lambda+1/2}(x)
\qquad (-1\leqslant x\leqslant1),
\label{3.41}
\end{equation}
Eq.\ (\ref{3.35}) is transformed into
\begin{equation}
G_{\nu}(\boldsymbol{r},\boldsymbol{r}')
=\frac{2^{N-4}\Gamma\left(\frac{N-1}{2}\right)}
{\pi^{(N-1)/2}\sin\left[\pi\left(\frac{N}{2}-\nu\right)\right]}
\frac{\displaystyle\rho^{N-2}C_{\nu-N/2+1}^{(N-1)/2}
\left(-1+\frac{\displaystyle2\rho^{2}
(\boldsymbol{r}-\boldsymbol{r}')^{2}}
{(r^{2}+\rho^{2})(r^{\prime\,2}+\rho^{2})}\right)}
{(r^{2}+\rho^{2})^{N/2-1}(r^{\prime\,2}+\rho^{2})^{N/2-1}}.
\label{3.42}
\end{equation}

Let us consider some particular cases. For $N=2$, from either of
Eqs.\ (\ref{3.35}), (\ref{3.39}) or (\ref{3.40}), with the aid of
Eq.\ (\ref{3.24}), we find
\begin{equation}
G_{\nu}(\boldsymbol{r},\boldsymbol{r}')
=\frac{1}{4\sin(\pi\nu)}P_{\nu}\left(-1+\frac{\displaystyle2\rho^{2}
(\boldsymbol{r}-\boldsymbol{r}')^{2}}
{(r^{2}+\rho^{2})(r^{\prime\,2}+\rho^{2})}\right)
\qquad (N=2).
\label{3.43}
\end{equation}
Next, it appears that the representations of
$G_{\nu}(\boldsymbol{r},\boldsymbol{r}')$ found above simplify
greatly when $N$ is odd, as then the Legendre function appearing in
Eqs.\ (\ref{3.35}), (\ref{3.39}) and (\ref{3.40}) may be expressed in
terms of trigonometric and inverse trigonometric functions either as
\cite[p.\ 168]{Magn66}
\begin{eqnarray}
P_{\nu}^{-N/2+1}(x) &=& \frac{\left(\frac{N-3}{2}\right)!}
{2^{N/2-2}\sqrt{\pi}}\left(1-x^{2}\right)^{-N/4+1/2}
\sum_{k=0}^{(N-3)/2}(-)^{k}
\frac{\Gamma\left(k+\nu-\frac{N}{2}+2\right)}
{k!\Gamma\left(k+\nu+\frac{3}{2}\right)
\left(\frac{N-3}{2}-k\right)!}
\nonumber \\
&& \times\sin\left[\left(2k+\nu-{\textstyle\frac{N}{2}}
+2\right)\arccos x\right]
\qquad (\textrm{$N$ odd, $N\geqslant3$})
\label{3.44}
\end{eqnarray}
or as \cite[p.\ 169]{Magn66}
\begin{eqnarray}
P_{\nu}^{-N/2+1}(x) &=& \sqrt{\frac{2}{\pi}}\,
\Gamma\left(\nu-{\textstyle\frac{N}{2}}+2\right)
\sum_{k=0}^{(N-3)/2}\frac{\left(k+\frac{N-3}{2}\right)!}
{2^{k}k!\Gamma\left(k+\nu+\frac{3}{2}\right)
\left(\frac{N-3}{2}-k\right)!}
\nonumber \\
&& \times\frac{\sin\left[\left(k+\nu+\frac{1}{2}\right)
\arccos x+\left(k-\frac{N-3}{2}\right)\frac{\pi}{2}\right]}
{(1-x^{2})^{k/2+1/4}}
\qquad (\textrm{$N$ odd, $N\geqslant3$}).
\nonumber \\
&&
\label{3.45}
\end{eqnarray}
In the simplest case of $N=3$, we have
\begin{equation}
P_{\nu}^{-1/2}(x)=\sqrt{\frac{2}{\pi}}\left(1-x^{2}\right)^{-1/4}
\frac{\sin\left[\left(\nu+\frac{1}{2}\right)\arccos x\right]}
{\nu+\frac{1}{2}}.
\label{3.46}
\end{equation}
Using this representation of $P_{\nu}^{-1/2}(x)$ in Eq.\
(\ref{3.39}), the latter being specialized to the case $N=3$, after
some straightforward movements and with the help of the identity
\begin{equation}
\Gamma\left({\textstyle\frac{1}{2}}+\nu\right)
\Gamma\left({\textstyle\frac{1}{2}}-\nu\right)
=\frac{\pi}{\cos(\pi\nu)},
\label{3.47}
\end{equation}
we arrive at
\begin{eqnarray}
G_{\nu}(\boldsymbol{r},\boldsymbol{r}')
&=& -\,\frac{1}{4\pi\cos(\pi\nu)}
\frac{\sqrt{(r^{2}+\rho^{2})(r^{\prime\,2}+\rho^{2})}}
{|\boldsymbol{r}-\boldsymbol{r}'|
\sqrt{r^{2}r^{\prime\,2}+2\rho^{2}\boldsymbol{r}\cdot\boldsymbol{r}'
+\rho^{4}}}
\nonumber \\
&& \times\sin\left[\left(\nu+\frac{1}{2}\right)
\arccos\left(-1+\frac{\displaystyle2\rho^{2}
(\boldsymbol{r}-\boldsymbol{r}')^{2}}
{(r^{2}+\rho^{2})(r^{\prime\,2}+\rho^{2})}\right)\right]
\qquad (N=3).
\label{3.48}
\end{eqnarray}
Equivalence between the result in Eq.\ (\ref{3.48}) and the following
expression (modified to conform with the present notation and
corrected for a sign error):
\begin{eqnarray}
G_{\nu}(\boldsymbol{r},\boldsymbol{r}')
&=& -\,\frac{1}{4\pi\cos(\pi\nu)}
\frac{\sqrt{(r^{2}+\rho^{2})(r^{\prime\,2}+\rho^{2})}}
{|\boldsymbol{r}-\boldsymbol{r}'|
\sqrt{r^{2}r^{\prime\,2}+2\rho^{2}\boldsymbol{r}\cdot\boldsymbol{r}'
+\rho^{4}}}
\nonumber \\
&& \times\sin\left[(2\nu+1)
\arctan\frac{\sqrt{r^{2}r^{\prime\,2}
+2\rho^{2}\boldsymbol{r}\cdot\boldsymbol{r}'+\rho^{4}}}
{\rho|\boldsymbol{r}-\boldsymbol{r}'|}\right]
\qquad (N=3),
\label{3.49}
\end{eqnarray}
given by Demkov and Ostrovsky in Ref.\ \cite{Demk71}, may be easily
established with the aid of the well-known inverse trigonometric
identity
\begin{equation}
\arctan\xi=\frac{1}{2}\arccos\frac{1-\xi^{2}}{1+\xi^{2}}
\qquad (\xi\geqslant0).
\label{3.50}
\end{equation}
\section{The generalized Green's function for the fish-eye problem}
\label{IV}
\setcounter{equation}{0}
A glance at either of Eqs.\ (\ref{3.35}), (\ref{3.39}) or
(\ref{3.40}) reveals that the fish-eye Green's function
$G_{\nu}(\boldsymbol{r},\boldsymbol{r}')$ fails to exist for the
following values of $\nu$:
\begin{equation}
\nu=n+N/2-1
\quad \textrm{or} \quad
\nu=-n-N/2
\qquad (n\in\mathbb{N}),
\label{4.1}
\end{equation}
which, by the way, are solutions of the quadratic equation
\begin{equation}
\nu(\nu+1)=\left(n+\frac{N}{2}\right)\left(n+\frac{N}{2}-1\right).
\label{4.2}
\end{equation}
If either of the conditions set in Eq.\ (\ref{4.1}) holds, one seeks
the generalized Green's function
$\bar{G}_{n+N/2-1}(\boldsymbol{r},\boldsymbol{r}')
\equiv\bar{G}_{-n-N/2}(\boldsymbol{r},\boldsymbol{r}')$,
defined through the limiting relation
\begin{eqnarray}
\bar{G}_{n+N/2-1}(\boldsymbol{r},\boldsymbol{r}')
&=& \lim\nolimits_{\nu(\nu+1)\to(n+N/2)(n+N/2-1)}
\nonumber \\
&& \frac{\partial}{\partial[\nu(\nu+1)]}
\left\{\left[\nu(\nu+1)-\left(n+\frac{N}{2}\right)
\left(n+\frac{N}{2}-1\right)\right]
G_{\nu}(\boldsymbol{r},\boldsymbol{r}')\right\}
\nonumber \\
&=& \frac{1}{2n+N-1}\lim_{\nu\to n+N/2-1}
\frac{\partial}{\partial\nu}
\left[\left(\nu-n-\frac{N}{2}+1\right)
\left(\nu+n+\frac{N}{2}\right)
G_{\nu}(\boldsymbol{r},\boldsymbol{r}')\right].
\nonumber \\
&&
\label{4.3}
\end{eqnarray}
If, for instance, representation (\ref{3.39}) of
$G_{\nu}(\boldsymbol{r},\boldsymbol{r}')$ is used in Eq.\
(\ref{4.3}), this results in
\begin{eqnarray}
\bar{G}_{n+N/2-1}(\boldsymbol{r},\boldsymbol{r}')
&=& (-)^{n}\frac{(n+N-2)!}{4\pi^{N/2}n!}
\frac{\displaystyle\rho^{N/2-1}}
{\displaystyle|\boldsymbol{r}-\boldsymbol{r}'|^{N/2-1}
\left(r^{2}r^{\prime\,2}+2\rho^{2}\boldsymbol{r}\cdot\boldsymbol{r}'
+\rho^{4}\right)^{N/4-1/2}}
\nonumber \\
&& \hspace*{-10em} 
\times\Bigg\{\left[\psi(n+N-1)-\psi(n+1)
+\frac{1}{2n+N-1}\right]P_{n+N/2-1}^{-N/2+1}
\left(-1+\frac{\displaystyle2\rho^{2}
(\boldsymbol{r}-\boldsymbol{r}')^{2}}
{(r^{2}+\rho^{2})(r^{\prime\,2}+\rho^{2})}\right)
\nonumber \\
&& \hspace*{-10em} 
\quad +\,\frac{\displaystyle\partial P_{\nu}^{-N/2+1}
\left(-1+\frac{\displaystyle2\rho^{2}
(\boldsymbol{r}-\boldsymbol{r}')^{2}}
{(r^{2}+\rho^{2})(r^{\prime\,2}+\rho^{2})}\right)}
{\partial\nu}\Bigg|_{\nu=n+N/2-1}\Bigg\},
\label{4.4}
\end{eqnarray}
where
\begin{equation}
\psi(\zeta)=\frac{1}{\Gamma(\zeta)}
\frac{\mathrm{d}\Gamma(\zeta)}{\mathrm{d}\zeta}
\label{4.5}
\end{equation}
is the digamma function. A number of closed-form representations of
the derivative $[\partial
P_{\nu}^{-N/2+1}(x)/\partial\nu]_{\nu=n+N/2-1}$ required in Eq.\
(\ref{4.4}) may be derived from the author's findings for
$[\partial P_{\nu}^{\pm m}(z)/\partial\nu]_{\nu=n}$,
$z\in\mathbb{C}\setminus(-1,1)$, presented in Refs.\
\cite{Szmy09a,Szmy09b}; the simplest, and thus potentially most
useful, of these expressions are listed in Appendix \ref{A}.

In the particular case of $N=3$, Eq.\ (\ref{4.4}) yields simply
\begin{eqnarray}
\bar{G}_{n+1/2}(\boldsymbol{r},\boldsymbol{r}')
&=& \frac{(-)^{n}}{4\pi^{2}}
\frac{\sqrt{(r^{2}+\rho^{2})(r^{\prime\,2}+\rho^{2})}}
{|\boldsymbol{r}-\boldsymbol{r}'|
\sqrt{r^{2}r^{\prime\,2}+2\rho^{2}\boldsymbol{r}\cdot\boldsymbol{r}'
+\rho^{4}}}
\nonumber \\
&& \hspace*{-5em}
\times\Bigg\{\cos\left[(n+1)
\arccos\left(-1+\frac{\displaystyle2\rho^{2}
(\boldsymbol{r}-\boldsymbol{r}')^{2}}
{(r^{2}+\rho^{2})(r^{\prime\,2}+\rho^{2})}\right)\right]
\arccos\left(-1+\frac{\displaystyle2\rho^{2}
(\boldsymbol{r}-\boldsymbol{r}')^{2}}
{(r^{2}+\rho^{2})(r^{\prime\,2}+\rho^{2})}\right)
\nonumber \\
&& \hspace*{-5em}
+\,\frac{\displaystyle\sin\left[(n+1)
\arccos\left(-1+\frac{\displaystyle2\rho^{2}
(\boldsymbol{r}-\boldsymbol{r}')^{2}}
{(r^{2}+\rho^{2})(r^{\prime\,2}+\rho^{2})}\right)\right]}
{2(n+1)}\Bigg\}
\qquad (N=3).
\label{4.6}
\end{eqnarray}
\section{Prospective applications}
\label{V}
\setcounter{equation}{0}
The closed-form representations of the fish-eye Green's function
$G_{\nu}(\boldsymbol{r},\boldsymbol{r}')$ and of its generalized
counterpart $\bar{G}_{n+N/2-1}(\boldsymbol{r},\boldsymbol{r}')$,
found in this work, are certainly interesting for their own
mathematical sake. They appear, however, to be also useful in the
physical context. In a forthcoming report, we shall use them to show
that, despite of claims to the contrary
\cite{Demk71,Leon09a,Leon09b}, in wave optics the infinite Maxwell
fish-eye medium does not possess the same perfect focusing properties
as it has in geometrical optics. Next, it has been confirmed
\cite{Biel10} that the use of either of the closed-form expressions
for $G_{\nu}(\boldsymbol{r},\boldsymbol{r}')$ listed in Section
\ref{III} simplifies greatly the mathematical analysis of
wave-optical properties of cylindrical ($N=2$) and spherical ($N=3$)
gradient-index lenses with the fish-eye refraction index (\ref{1.1})
and of finite radii $r_{\mathrm{lens}}\leqslant\rho\sqrt{2n_{0}-1}$.
Finally, in yet another forthcoming paper, we shall show that there
is a close mathematical relationship between the wavized Maxwell
fish-eye problem in $\mathbb{R}^{N}$ and the $N$-dimensional
Schr{\"o}dinger--Coulomb problem in momentum space; in particular, we
shall provide there an integral expression for the momentum-space
Schr{\"o}dinger--Coulomb Green's function in terms of the fish-eye
Green's function discussed above.
\appendix
\section{Appendix: The derivatives $[\partial
P_{\nu}^{-N/2+1}(x)/\partial\nu]_{\nu=n+N/2-1}$ for
$N\in\mathbb{N}\setminus\{0,1\}$}
\label{A}
\setcounter{equation}{0}
From the relations
\begin{equation}
P_{n}^{-m}(x)=\mathrm{e}^{-\mathrm{i}\pi m/2}\frac{(n-m)!}{(n+m)!}
P_{n}^{m}(x+\mathrm{i}0)
\qquad (0\leqslant m\leqslant n)
\label{A.1}
\end{equation}
and
\begin{eqnarray}
\frac{\partial P_{\nu}^{-m}(x)}{\partial\nu}\bigg|_{\nu=n}
&=& \mathrm{e}^{-\mathrm{i}\pi m/2}\frac{(n-m)!}{(n+m)!}
\frac{\partial P_{n}^{m}(x+\mathrm{i}0)}{\partial\nu}\bigg|_{\nu=n}
-[\psi(n+m+1)-\psi(n-m+1)]P_{n}^{-m}(x)
\nonumber \\
&& \hspace*{20em} (0\leqslant m\leqslant n),
\label{A.2}
\end{eqnarray}
and from a number of closed-form expressions for the derivative
$[\partial P_{\nu}^{m}(z)/\partial\nu]_{\nu=n}$, with
$z\in\mathbb{C}\setminus(-1,1)$ and $0\leqslant m\leqslant n$, found
by the present author in Refs.\ \cite{Szmy09a,Szmy09b}, one may
derive, among others, the following representations of $[\partial
P_{\nu}^{-N/2+1}(x)/\partial\nu]_{\nu=n+N/2-1}$, with
$x\in[-1,1]$ and with \emph{$N$ being an even natural number greater
than zero}:\footnote{Attention! The \emph{Mathematica\/} 7.0.0
function \texttt{LegendreP[nu,mu,2,x]} evaluates incorrectly (the
modulus is correct but the sign is wrong) numerical values of the
associated Legendre functions $P_{2n+1}^{-2n-1}(x)$ and
$P_{2n+2}^{-2n-1}(x)$ for $n\in\mathbb{N}$, $-1\leqslant
x\leqslant1$. An empirically discovered remedy is to subject the
variable $x$ to the action of the function \texttt{SetPrecision}
before it is used as an argument of \texttt{LegendreP}.}
\begin{eqnarray}
\frac{\partial P_{\nu}^{-N/2+1}(x)}
{\partial\nu}\Bigg|_{\nu=n+N/2-1}
&=& P_{n+N/2-1}^{-N/2+1}(x)\ln\frac{1+x}{2}
-2\psi(n+N-1)P_{n+N/2-1}^{-N/2+1}(x)
\nonumber \\
&& \hspace*{-10em}
+\,\frac{n!}{(n+N-2)!}\left(\frac{1-x^{2}}{4}\right)^{N/4-1/2}
\sum_{k=0}^{n}(-)^{k}\frac{(k+n+N-2)!}{k!(k+N/2-1)!(n-k)!}
\nonumber \\
&& \hspace*{-10em}
\quad\times[2\psi(k+n+N-1)-\psi(k+N/2)]
\left(\frac{1-x}{2}\right)^{k}
\nonumber \\
&& \hspace*{-10em}
+\,\left(\frac{1-x}{1+x}\right)^{N/4-1/2}
\sum_{k=0}^{n+N/2-1}(-)^{k}\frac{(k+n+N/2-1)!\psi(k+N/2)}
{k!(k+N/2-1)!(n+N/2-k-1)!}\left(\frac{1-x}{2}\right)^{k},
\nonumber \\
&&
\label{A.3}
\end{eqnarray}
\begin{eqnarray}
\frac{\partial P_{\nu}^{-N/2+1}(x)}
{\partial\nu}\Bigg|_{\nu=n+N/2-1}
&=& P_{n+N/2-1}^{-N/2+1}(x)\ln\frac{1+x}{2}
-2\psi(n+N/2)P_{n+N/2-1}^{-N/2+1}(x)
\nonumber \\
&& \hspace*{-10em}
+\,\frac{n!}{(n+N-2)!}\left(\frac{1-x^{2}}{4}\right)^{N/4-1/2}
\sum_{k=0}^{n}(-)^{k}\frac{(k+n+N-2)!\psi(k+N/2)}
{k!(k+N/2-1)!(n-k)!}\left(\frac{1-x}{2}\right)^{k}
\nonumber \\
&& \hspace*{-10em}
+\,\left(\frac{1-x}{1+x}\right)^{N/4-1/2}
\sum_{k=0}^{n+N/2-1}(-)^{k}\frac{(k+n+N/2-1)!}
{k!(k+N/2-1)!(n+N/2-k-1)!}
\nonumber \\
&& \hspace*{-10em}
\quad\times[2\psi(k+n+N/2)-\psi(k+N/2)]
\left(\frac{1-x}{2}\right)^{k},
\label{A.4}
\end{eqnarray}
\begin{eqnarray}
\frac{\partial P_{\nu}^{-N/2+1}(x)}
{\partial\nu}\Bigg|_{\nu=n+N/2-1}
&=& P_{n+N/2-1}^{-N/2+1}(x)\ln\frac{1+x}{2}
\nonumber \\
&& \hspace*{-10em}
-\,[\psi(n+N-1)+\psi(n+N/2)]P_{n+N/2-1}^{-N/2+1}(x)
\nonumber \\
&& \hspace*{-10em}
+\,\frac{n!}{(n+N-2)!}\left(\frac{1-x^{2}}{4}\right)^{N/4-1/2}
\sum_{k=0}^{n}(-)^{k}\frac{(k+n+N-2)!\psi(k+n+N-1)}
{k!(k+N/2-1)!(n-k)!}\left(\frac{1-x}{2}\right)^{k}
\nonumber \\
&& \hspace*{-10em}
+\,\left(\frac{1-x}{1+x}\right)^{N/4-1/2}
\sum_{k=0}^{n+N/2-1}(-)^{k}\frac{(k+n+N/2-1)!\psi(k+n+N/2)}
{k!(k+N/2-1)!(n+N/2-k-1)!}\left(\frac{1-x}{2}\right)^{k},
\label{A.5}
\end{eqnarray}
\begin{eqnarray}
\frac{\partial P_{\nu}^{-N/2+1}(x)}
{\partial\nu}\Bigg|_{\nu=n+N/2-1} 
&=& P_{n+N/2-1}^{-N/2+1}(x)\ln\frac{1+x}{2}
\nonumber \\
&& \hspace*{-10em} 
-\,(-)^{n}\left(\frac{1-x^{2}}{4}\right)^{-N/4+1/2}
\sum_{k=0}^{N/2-2}\frac{(k+n)!(N/2-k-2)!}
{k!(n+N-k-2)!}\left(\frac{1+x}{2}\right)^{k}
\nonumber \\
&& \hspace*{-10em} 
+\,(-)^{n}\left(\frac{1+x}{1-x}\right)^{N/4-1/2}
\sum_{k=0}^{n+N/2-1}(-)^{k}\frac{(k+n+N/2-1)!}
{k!(k+N/2-1)!(n+N/2-k-1)!}
\nonumber \\
&& \hspace*{-10em}
\quad\times[2\psi(k+n+N/2)-\psi(k+N/2)-\psi(k+1)]
\left(\frac{1+x}{2}\right)^{k},
\label{A.6}
\end{eqnarray}
\begin{eqnarray}
\frac{\partial P_{\nu}^{-N/2+1}(x)}
{\partial\nu}\Bigg|_{\nu=n+N/2-1}
&=& P_{n+N/2-1}^{-N/2+1}(x)\ln\frac{1+x}{2}
\nonumber \\
&& \hspace*{-10em}
-\,[\psi(n+N-1)-\psi(n+1)]P_{n+N/2-1}^{-N/2+1}(x)
\nonumber \\
&& \hspace*{-10em} 
-\,(-)^{n}\frac{n!}{(n+N-2)!}
\left(\frac{1-x}{1+x}\right)^{N/4-1/2}
\sum_{k=0}^{N/2-2}\frac{(k+n+N/2-1)!(N/2-k-2)!}{k!(n+N/2-k-1)!}
\left(\frac{1+x}{2}\right)^{k}
\nonumber \\
&& \hspace*{-10em}
+\,(-)^{n}\frac{n!}{(n+N-2)!}
\left(\frac{1-x^{2}}{4}\right)^{N/4-1/2}
\sum_{k=0}^{n}(-)^{k}\frac{(k+n+N-2)!}{k!(k+N/2-1)!(n-k)!}
\nonumber \\
&& \hspace*{-10em} 
\quad\times[2\psi(k+n+N-1)-\psi(k+N/2)-\psi(k+1)]
\left(\frac{1+x}{2}\right)^{k},
\label{A.7}
\end{eqnarray}
\begin{eqnarray}
\frac{\partial P_{\nu}^{-N/2+1}(x)}
{\partial\nu}\Bigg|_{\nu=n+N/2-1}
&=& P_{n+N/2-1}^{-N/2+1}(x)\ln\frac{1+x}{2}
\nonumber \\
&& \hspace*{-10em}
+\,[\psi(n+1)+\psi(n+N/2)]P_{n+N/2-1}^{-N/2+1}(x)
\nonumber \\
&& \hspace*{-10em}
-\,(-)^{n}n!(n+N/2-1)!\left(\frac{1+x}{1-x}\right)^{N/4-1/2}
\left(\frac{1-x}{2}\right)^{n+N/2-1}
\nonumber \\
&& \hspace*{-10em}
\quad\times\sum_{k=1}^{N/2-1}\frac{(k-1)!}
{(k+n)!(k+n+N/2-1)!(N/2-k-1)!}\left(\frac{1-x}{1+x}\right)^{k}
\nonumber \\
&& \hspace*{-10em}
-\,n!(n+N/2-1)!\left(\frac{1-x}{1+x}\right)^{N/4-1/2}
\left(\frac{1+x}{2}\right)^{n+N/2-1}
\nonumber \\
&& \hspace*{-10em}
\quad\times\sum_{k=0}^{n}(-)^{k}\frac{\psi(n+N/2-k)+\psi(n-k+1)}
{k!(k+N/2-1)!(n-k)!(n+N/2-k-1)!}\left(\frac{1-x}{1+x}\right)^{k},
\label{A.8}
\end{eqnarray}
\begin{eqnarray}
\frac{\partial P_{\nu}^{-N/2+1}(x)}
{\partial\nu}\Bigg|_{\nu=n+N/2-1}
&=& P_{n+N/2-1}^{-N/2+1}(x)\ln\frac{1+x}{2}
\nonumber \\
&& \hspace*{-10em}
+\,[\psi(n+1)+\psi(n+N/2)]P_{n+N/2-1}^{-N/2+1}(x)
\nonumber \\
&& \hspace*{-10em}
-\,(-)^{n}n!(n+N/2-1)!\left(\frac{1-x}{1+x}\right)^{N/4-1/2}
\left(\frac{1-x}{2}\right)^{n+N/2-1}
\nonumber \\
&& \hspace*{-10em}
\quad\times\sum_{k=0}^{N/2-2}\frac{(N/2-k-2)!}
{k!(n+N/2-k-1)!(n+N-k-2)!}\left(\frac{1+x}{1-x}\right)^{k}
\nonumber \\
&& \hspace*{-10em}
-\,(-)^{n}n!(n+N/2-1)!\left(\frac{1+x}{1-x}\right)^{N/4-1/2}
\left(\frac{1-x}{2}\right)^{n+N/2-1}
\nonumber \\
&& \hspace*{-10em}
\quad\times\sum_{k=0}^{n}(-)^{k}\frac{\psi(k+1)+\psi(k+N/2)}
{k!(k+N/2-1)!(n-k)!(n+N/2-k-1)!}\left(\frac{1+x}{1-x}\right)^{k},
\label{A.9}
\end{eqnarray}
\begin{eqnarray}
\frac{\partial P_{\nu}^{-N/2+1}(x)}
{\partial\nu}\Bigg|_{\nu=n+N/2-1}
&=& P_{n+N/2-1}^{-N/2+1}(x)\ln\frac{1+x}{2}
\nonumber \\
&& \hspace*{-10em}
+\,[2\psi(2n+N-1)-\psi(n+N-1)-\psi(n+N/2)]P_{n+N/2-1}^{-N/2+1}(x)
\nonumber \\
&& \hspace*{-10em}
+\sum_{k=0}^{n-1}(-)^{k+n}\frac{2k+N-1}{(n-k)(k+n+N-1)}
\left[1+\frac{n!(k+N-2)!}{k!(n+N-2)!}\right]
P_{k+N/2-1}^{-N/2+1}(x)
\nonumber \\
&& \hspace*{-10em}
-\sum_{k=0}^{N/2-2}(-)^{k+n+N/2}\frac{2k+1}{(n+N/2-k-1)(k+n+N/2)}
P_{k}^{-N/2+1}(x).
\label{A.10}
\end{eqnarray}
If, in turn, \emph{$N$ is an odd natural number greater than 1},
then from Eqs.\ (\ref{3.44}) and (\ref{3.45}) one obtains
\begin{eqnarray}
\frac{\partial P_{\nu}^{-N/2+1}(x)}
{\partial\nu}\Bigg|_{\nu=n+N/2-1}
&=& \frac{2}{\pi}Q_{n+N/2-1}^{-N/2+1}(x)\arccos x
\nonumber \\
&& \hspace*{-10em}
-\,\frac{\left(\frac{N-3}{2}\right)!}
{2^{N/2-2}\sqrt{\pi}}\left(1-x^{2}\right)^{-N/4+1/2}
\sum_{k=0}^{(N-3)/2}(-)^{k}\frac{(k+n)!}
{k!\left(k+n+\frac{N-1}{2}\right)!\left(\frac{N-3}{2}-k\right)!}
\nonumber \\
&& \hspace*{-10em}
\quad\times\left[\psi\left(k+n+{\textstyle\frac{N+1}{2}}\right)
-\psi(k+n+1)\right]\sin[(2k+n+1)\arccos x]
\label{A.11}
\end{eqnarray}
and
\begin{eqnarray}
\frac{\partial P_{\nu}^{-N/2+1}(x)}
{\partial\nu}\Bigg|_{\nu=n+N/2-1}
&=& \frac{2}{\pi}Q_{n+N/2-1}^{-N/2+1}(x)\arccos x
+\psi(n+1)P_{n+N/2-1}^{-N/2+1}(x)
\nonumber \\
&& \hspace*{-10em}
-\,\sqrt{\frac{2}{\pi}}\,n!
\sum_{k=0}^{(N-3)/2}\frac{\left(k+\frac{N-3}{2}\right)!
\psi\left(k+n+\frac{N+1}{2}\right)}
{2^{k}k!\left(k+n+\frac{N-1}{2}\right)!
\left(\frac{N-3}{2}-k\right)!}
\nonumber \\
&& \hspace*{-10em}
\quad\times\frac{\sin\left[\left(k+n+\frac{N-1}{2}\right)
\arccos x+\left(k-\frac{N-3}{2}\right)\frac{\pi}{2}\right]}
{(1-x^{2})^{k/2+1/4}},
\label{A.12}
\end{eqnarray}
where
\begin{eqnarray}
P_{n+N/2-1}^{-N/2+1}(x) 
&=& \frac{\left(\frac{N-3}{2}\right)!}
{2^{N/2-2}\sqrt{\pi}}\left(1-x^{2}\right)^{-N/4+1/2}
\sum_{k=0}^{(N-3)/2}(-)^{k}
\frac{(k+n)!}{k!\left(k+n+\frac{N-1}{2}\right)!
\left(\frac{N-3}{2}-k\right)!}
\nonumber \\
&& \times\sin[(2k+n+1)\arccos x]
\label{A.13}
\end{eqnarray}
or, equivalently,
\begin{eqnarray}
P_{n+N/2-1}^{-N/2+1}(x) &=& \sqrt{\frac{2}{\pi}}\,n!
\sum_{k=0}^{(N-3)/2}\frac{\left(k+\frac{N-3}{2}\right)!}
{2^{k}k!\left(k+n+\frac{N-1}{2}\right)!
\left(\frac{N-3}{2}-k\right)!}
\nonumber \\
&& \times\frac{\sin\left[\left(k+n+\frac{N-1}{2}\right)
\arccos x+\left(k-\frac{N-3}{2}\right)\frac{\pi}{2}\right]}
{(1-x^{2})^{k/2+1/4}}
\label{A.14}
\end{eqnarray}
and also
\begin{eqnarray}
Q_{n+N/2-1}^{-N/2+1}(x) 
&=& \frac{\sqrt{\pi}\left(\frac{N-3}{2}\right)!}
{2^{N/2-1}}\left(1-x^{2}\right)^{-N/4+1/2}
\sum_{k=0}^{(N-3)/2}(-)^{k}
\frac{(k+n)!}{k!\left(k+n+\frac{N-1}{2}\right)!
\left(\frac{N-3}{2}-k\right)!}
\nonumber \\
&& \times\cos[(2k+n+1)\arccos x]
\label{A.15}
\end{eqnarray}
or, equivalently,
\begin{eqnarray}
Q_{n+N/2-1}^{-N/2+1}(x) &=& \sqrt{\frac{\pi}{2}}\,n!
\sum_{k=0}^{(N-3)/2}\frac{\left(k+\frac{N-3}{2}\right)!}
{2^{k}k!\left(k+n+\frac{N-1}{2}\right)!
\left(\frac{N-3}{2}-k\right)!}
\nonumber \\
&& \times\frac{\cos\left[\left(k+n+\frac{N-1}{2}\right)
\arccos x+\left(k-\frac{N-3}{2}\right)\frac{\pi}{2}\right]}
{(1-x^{2})^{k/2+1/4}}.
\label{A.16}
\end{eqnarray}
%
%\newpage
%

%

\begin{thebibliography}{99}
\bibitem{Maxw54}
   J.\ Clerk Maxwell,
   Solutions of problems,
   Camb.\ Dublin Math.\ J.\ 8 (1854) 188
   [reprinted in: \emph{The Scientific Papers of James Clerk
   Maxwell\/} (Cambridge University Press, Cambridge, 1890; 
   Dover, New York, 1965), p.\ 74], the solution to problem 2.
\bibitem{Cara26}
   C.\ Carath{\'e}odory,
   {\"U}ber den Zusammenhang der Theorie der absoluten optischen
   Instrumente mit einem Satze der Variationsrechnung,
   Sitzungsberichte der Bayerischen Akademie der Wissenschaften.\
   Mathematisch-naturwissenschaftliche Abteilung (1926) 1--18
   [reprinted in:\ 
   C.\ Carath{\'e}odory, 
   Gesammelte mathematische Schriften, Band II
   (Beck, Munich, 1955), pp.\ 181--97], sections 3 and 4
\bibitem{Wolf04}
   K.\ B.\ Wolf,
   Geometric Optics on Phase Space
   (Springer, Berlin, 2004), chapter 6
\bibitem{Demk71}
   Yu.\ N.\ Demkov, V.\ N.\ Ostrovsky,
   Intrinsic symmetry of the Maxwell `fish-eye' problem and the Fock
   group for the hydrogen atom,
   Zh.\ Eksp.\ Teor.\ Fiz.\ 60 (1971) 2011
   [Sov.\ Phys.\ -- JETP 33 (1971) 1083]
\bibitem{Fran90}
   A.\ Frank, F.\ Leyvraz, K.\ B.\ Wolf,
   Hidden symmetry and potential group of the Maxwell fish-eye,
   J.\ Math.\ Phys.\ 31 (1990) 2757
\bibitem{Fran91}
   A.\ Frank, F.\ Leyvraz, K.\ B.\ Wolf,
   Potential group in optics:\ the Maxwell fish-eye system,
   in:\ Group Theoretical Methods in Physics,
   V.\ V.\ Dodonov, V.\ I.\ Man'ko (eds.),
   Lecture Notes in Physics 382
   (Springer, Berlin, 1991), p.\ 111
\bibitem{Mako09}
   A.\ J.\ Makowski, K.\ J.\ G{\'o}rska,
   Quantization of the Maxwell fish-eye problem and the
   quantum--classical correspondence,
   Phys.\ Rev.\ A 79 (2009) 052116
\bibitem{Leon09a}
   U.\ Leonhardt,
   Perfect imaging without negative refraction,
   New J.\ Phys.\ 11 (2009) 093040
\bibitem{Leon09b}
   U.\ Leonhardt, T.\ G.\ Philbin,
   Perfect imaging with positive refraction in three dimensions,
   preprint arXiv:0911.0552v1
\bibitem{Thom72}
   W.\ Thomson (Lord Kelvin),
   Reprint of Papers on Electrostatics and Magnetism,
   Macmillan, London, 1872, Article XIV.
\bibitem{Demk84}
   Yu.\ N.\ Demkov, N.\ V.\ Semenova,
   Inversion transformation in the Schr{\"o}dinger equation,
   Teor.\ Mat.\ Fiz.\ 60 (1984) 423 
   [Theor.\ Math.\ Phys.\ 60 (1984) 914]
\bibitem{Magn66}
   W.\ Magnus, F.\ Oberhettinger, R.\ P.\ Soni,
   Formulas and Theorems for the Special Functions of Mathematical
   Physics, 3rd ed.\
   (Springer, Berlin, 1966)
\bibitem{Szmy09a}
   R.\ Szmytkowski,
   On the derivative of the associated Legendre function of the first
   kind of integer order with respect to its degree (with
   applications to the construction of the associated Legendre
   function of the second kind of integer degree and order),
   preprint arXiv:0907.3217
\bibitem{Szmy09b}
   R.\ Szmytkowski,
   On parameter derivatives of the associated Legendre function of the 
   first kind (with applications to the construction of the 
   associated Legendre function of the second kind of integer degree 
   and order),
   preprint arXiv:0910.4550
\bibitem{Biel10}
   S.\ Bielski, private communication
\end{thebibliography}
\end{document}